\documentclass[sigconf]{acmart}
\usepackage{graphicx}
\usepackage[ruled]{algorithm2e}
\usepackage[utf8]{inputenc}
\usepackage{subfigure}
\usepackage{multirow}
\usepackage{tabulary}
\usepackage{amsthm}
\usepackage{algorithmic}
\usepackage{enumitem}

\AtBeginDocument{%
  \providecommand\BibTeX{{%
    \normalfont B\kern-0.5em{\scshape i\kern-0.25em b}\kern-0.8em\TeX}}}

\begin{document}

\title{Graph Embedding for Recommendation against Attribute Inference Attacks}
\author{Shijie Zhang}
\affiliation{%
  \institution{The University of Queensland}}
\email{shijie.zhang@uq.edu.au}

\author{Hongzhi Yin}
\authornote{Corresponding author; contributing equally with the first author.}
\affiliation{%
  \institution{The University of Queensland}}
\email{h.yin1@uq.edu.au}
\author{Tong Chen}
\affiliation{%
  \institution{The University of Queensland}}
\email{tong.chen@uq.edu.au}
\author{Zi Huang}
\affiliation{%
  \institution{The University of Queensland}}
\email{huang@itee@uq.edu.au}
\author{Lizhen Cui}
\affiliation{%
  \institution{Shandong University}}
\email{clz@sdu.edu.cn}
\author{Xiangliang Zhang}
\affiliation{%
  \institution{King Abdullah University of Science and Technology}}
\email{xiangliang.zhang@kaust.edu.sa}

\begin{abstract}
In recent years, recommender systems play a pivotal role in helping users identify the most suitable items that satisfy personal preferences. As user-item interactions can be naturally modelled as graph-structured data, variants of graph convolutional networks (GCNs) have become a well-established building block in the latest recommenders. Due to the wide utilization of sensitive user profile data, existing recommendation paradigms are likely to expose users to the threat of privacy breach, and GCN-based recommenders are no exception. Apart from the leakage of raw user data, the fragility of current recommenders under inference attacks offers malicious attackers a backdoor to estimate users' private attributes via their behavioral footprints and the recommendation results. However, little attention has been paid to developing recommender systems that can defend such attribute inference attacks, and existing works achieve attack resistance by either sacrificing considerable recommendation accuracy or only covering specific attack models or protected information. In our paper, we propose GERAI, a novel differentially private graph convolutional network to address such limitations. Specifically, in GERAI, we bind the information perturbation mechanism in differential privacy with the recommendation capability of graph convolutional networks. Furthermore, based on local differential privacy and functional mechanism, we innovatively devise a dual-stage encryption paradigm to simultaneously enforce privacy guarantee on users' sensitive features and the model optimization process. Extensive experiments show the superiority of GERAI in terms of its resistance to attribute inference attacks and recommendation effectiveness. 
\end{abstract}

\ccsdesc[500]{Information systems~Collaborative filtering}
\keywords{Privacy-preserving Recommender System; Attribute Inference Attacks; Deep Learning; Differential Privacy}


\maketitle
\section{Introduction}
With the explosive growth of e-commerce, consumers are shopping with online platforms more frequently~\cite{guo2019streaming, chen2019air, yin2019social}. As an effective solution to information overload, recommender systems automatically discover the most relevant items or services for each user and thus improve both the user experience and business revenue. For this reason, recommender systems have become an indispensable part in our contemporary lives. 

Latent factor models like matrix factorization~\cite{mnih2008probabilistic} are typical collaborative filtering-based recommendations, which infer user-item interactions via learned latent user/item representations. Because user-item interactions can be conveniently formulated as graph-structured data, graph embedding-based recommenders~\cite{xie2016learning,shi2018heterogeneous,zhang2019inferring} are highly effective in uncovering users' subtle preferences toward items. As deep neural networks demonstrate superior capability of representation learning in various machine learning tasks, deep recommendation models, especially those derived from graph convolutional networks (GCNs) \cite{ying2018graph,wang2019neural,xia2020self,yu2020enhance} have recently become one of the most prominent techniques in this field.



To enhance the recommendation performance, especially for fresh (i.e., cold-start) customers, it is a common practice to incorporate side information (a.k.a. features or contexts) \cite{adomavicius2011context,zhang2020gcn,wang2020next} about users.  
During user registration, some service providers even start persuading users to complete questionnaires about personal demographics to facilitate user profiling. 
However, the utilization of user data containing personal information often sparks serious privacy concerns. A 2018 survey \cite{kats2018many} showed that more than $80\%$ US Internet users were concerned about how their personal data is being used on Facebook; and among Facebook users sharing less content on social media, $47\%$ reported that privacy issue was the main concern. Consequently, with the growing public awareness on privacy, a dilemma is presented to e-commerce platforms: either they proceed with such sensitive data acquisition process despite the high risk on privacy breach, or they allow users not to disclose their sensitive attributes but provide compromised recommendation performance as a result. In that sense, a sound privacy guarantee on the user side is highly desirable, which avoids uploading the unencrypted raw user features to a recommender system. Furthermore, according to the example that Apple is now telling users their personal data is protected before being shared for analytics, it also helps increase users' willingness to share their sensitive data.


Meanwhile, a more critical privacy issue comes from the fact that  users' sensitive attributes can still be disclosed purely based on how they behave. Regardless of the availability of features, recommenders learn explicit or latent profiles that reflect users' preferences based on her/his behavioral footprints (e.g., previous ratings and reviews), and produce personalized recommendations with the constructed profiles~\cite{rashid2002getting}. However, many early studies have shown that even a user's personal information can be accurately inferred via her/his interaction history~\cite{weinsberg2012blurme,kosinski2013private,calandrino2011you}. Such personal information includes age, gender, political orientation, health, financial status etc. and are highly confidential. Furthermore, the inferred attributes can be utilized to link users across multiple sites and break anonymity~\cite{shu2017user,goga2013exploiting}. For example, \cite{naranyanan2008robust} successfully deanonymizes Netflix users using the public IMDb user profiles. Due to the open-access nature of many platforms (e.g., Yelp and Amazon), users' behavioral trajectories can be easily captured by a malicious third-party, leading to catastrophic leakage of inferred user attributes. This is known as the attribute inference attack~\cite{gong2016you}, where the malicious attackers can be cyber criminals, data brokers, advertisers, etc. 
By proving that even a person's racial information and sexual orientation can be precisely predicted from merely the ``like'' behaviors on Facebook, Kosinski et al.~\cite{kosinski2013private} demonstrated that users' preference signals are highly vulnerable to attribute inference attacks. This is especially alarming for many GCN-based recommenders, since user representations are usually formed by aggregating information from her/his interacted items. Moreover, the personalized recommendation results can also be utilized by attackers since they are strong reflections on users' preferences and are increasingly accessible via services like friend activity tracing (e.g., Spotify) and group recommendation~\cite{yin2019social}.
Hence, this motivates us to design a secure recommender system that stays robust against attribute inference attacks.

In GCN-based recommenders, graphs are constructed by linking user and item nodes via their interactions. However, though existing GCNs are advantageous in binding a node's own features and its high-order connectivity with other nodes into an expressive representation, they exhibit very little consideration on user privacy. In fact, the field of privacy-preserving recommender systems that are resistant to attribute inference attacks is far from its maturity. \cite{polat2005privacy,nikolaenko2013privacy,erkin2010privacy,canny2002collaborative} have applied cryptography algorithms to the recommendation models, but the computational cost of encryption is too high to support real-world deployment. Recently, the notion of differential privacy (DP) has become a well-established approach for protecting the confidentiality of personal data. Essentially, DP works by adding noise to each data instance (i.e., perturbation), thus masking the original information in the data. In the context of both recommendation and graph embedding, there has also been attempts to adopt DP to perturb the output of matrix factorization algorithms~\cite{liu2015fast,berlioz2015applying,xu2018dpne}. Unfortunately, these approaches are designed to only prevent membership attacks which infer users' real ratings in the dataset, and are unable to provide a higher level of protection on users' sensitive information against inference attacks. A recent work~\cite{beigi2020privacy} systematically investigates the problem of developing and evaluating recommender systems under the attribute inference attack setting. Their proposed model RAP~\cite{beigi2020privacy} utilizes an adversarial learning paradigm where a personalized recommendation model and an attribute inference attack model are trained against each other, hence the attackers are more likely to fail when inferring user attributes from interaction records. However, it suffers from two major limitations. Firstly, as the design of RAP requires a pre-specified and fixed attribute inference model, its resistance to any arbitrary attacker is unguaranteed given the unpredictability of the inference model that an attacker may choose. 
Secondly, though RAP assumes the existence of users' sensitive attributes, it only treats them as ground-truth labels for training the inference model, and does not incorporate such important side information for recommendation. This design not only fails to ease users' privacy concerns on submitting their original attributes, but also greatly hinders the model's ability to securely utilize user features to achieve more accurate recommendation results. 

To this end, we address a largely overlooked defect of existing GCN-based recommenders, i.e., protecting users' private attributes from attribute inference attacks. 
Meanwhile, unlike existing inference-resistant recommenders, we would like the model to take advantage of user information for accurate recommendation without exerting privacy breach. 
In this paper, we subsume the GCN-based recommender under the differential privacy (DP) constraint, and propose a novel privacy-preserving recommender GERAI, namely Graph Embedding for Recommendation against Attribute Inference Attacks. In GERAI, we build its recommendation module upon the state-of-the-art inductive GCNs~\cite{defferrard2016convolutional,hamilton2017inductive,kipf2016semi} to jointly exploit the user-item interactions and the rich side information of users. To achieve optimal privacy strength, we propose a novel dual-stage perturbation paradigm with DP. Firstly, at the input stage, GERAI performs perturbation on the raw user features. On one hand, this offers users a privacy guarantee while sharing their sensitive data. On the other hand, the perturbed user features will make the generated recommendations less dependent on a user's true attributes, making it harder to infer those attributes via recommendation results. Specifically, we introduce local differential privacy (LDP) for feature perturbation, where each individual's original feature vector is transformed into a noisy version before being processed by the recommendation module. We further demonstrate that the perturbed input data satisfies the LDP constraint while retaining adequate utility for the recommender to learn the subtle user preferences. Secondly, we enforce DP on the optimization stage of GERAI so that the recommendation results are less likely to reveal a user's attributes and preferences~\cite{liu2015fast,berlioz2015applying,beigi2020privacy} in the inference attack. 
To achieve this, we innovatively resort to the functional mechanism~\cite{zhang2012functional} that allows to enforce DP by perturbing the loss function in the learning process. Different from methods that applies perturbation on recommendation results \cite{berlioz2015applying}, by perturbing the loss function, GERAI defends the inference attack without setting obstacles for learning meaningful associations between user profiles and recommended items. 

Overall, we summarize our contributions in the following:
\begin{itemize}
    \item We address the increasing privacy concerns in the recommendation context, and propose a novel solution GERAI, namely differentially private graph convolutional network to protect users' sensitive data against attribute inference attacks and provide high-quality recommendations at the same time.
    \item Our proposed GERAI innovatively incorporates differential privacy with a dual-stage perturbation strategy for both the input features and the optimization process. As such, GERAI assures user privacy and offers better recommendation effectiveness than existing privacy-preserving recommenders.
    \item We conduct extensive experiments to evaluate the performance of GERAI on real-world data. Comparisons with state-of-the-art baselines show that GERAI provides a better privacy guarantee with less compromise on the recommendation accuracy.
\end{itemize}

\section{Preliminaries}
In this section, we first revisit the definitions of differential privacy and then formally define our problem. Note that in the description below, all vectors and matrices are respectively denoted wiht bold lowercase and bold uppercase letters, and all sets are written in calligraphic uppercase letters.

\textbf{Differential Privacy.} Differential privacy (DP) is a strong mathematical guarantee of privacy in the context of machine learning tasks. DP was first introduced by~\cite{dwork2014algorithmic} and it aims to preclude adversarial inference on any raw input data from a model's output. Given a privacy coefficient $\epsilon > 0$, the $\epsilon-$differential privacy ($\epsilon-$DP) is defined as follows:

\textit{Definition 2.1}. ($\epsilon-$Differential Privacy) For a randomized function (e.g., a perturbation algorithm or machine learning model) $f(\cdot)$ that takes a dataset as its input, it satisfies $\epsilon-$DP if: 
\begin{equation}\label{eq:DP}
    Pr[f(\mathcal{D}) \in O] \leq exp(\epsilon)Pr[f(\mathcal{D}')\in O],
\end{equation}
where $Pr[\cdot]$ represents probability, $\mathcal{D}$ and $\mathcal{D}^{'}$ are any two datasets differing on only one data instance, and $O$ denotes all subsets of possible output values that $f(\cdot)$ produces. If $O$ is continuous, then the probability term can be replaced by a probability density function. Eq.(\ref{eq:DP}) implies that the probability of generating the model output with $\mathcal{D}$ is at most $exp(\epsilon)$ times smaller than with $\mathcal{D}'$. That is, $f(\cdot)$ should not overly depend on any individual data instance, providing each instance roughly the same privacy. 
As a common practice for privacy protection, each individual user's personal data can be perturbed by adding controlled noise before it is fed into $f(\cdot)$. 
In this case, the data owned by every user is regarded as a singleton dataset, and we require the function $f(\cdot)$ to provide differential privacy when such a singleton database is given as the input. Specifically, this is termed as $\epsilon-$local differential privacy ($\epsilon-$LDP): 

\textit{Definition 2.1}. ($\epsilon-$Local Differential Privacy) A randomized function $f(\cdot)$ satisfies $\epsilon-$LDP if and only if for any two users' data $t$ and $t'$, we have:
\begin{equation}\label{eq:LDP}
    Pr[f(t) = t^{*}] \leq exp(\epsilon) \cdot Pr[f(t') = t^{*}]
\end{equation}
where $t^{*}$ denotes the output of $f(\cdot)$. The lower $\epsilon$ provides stronger privacy but may result in lower accuracy of a trained machine learning model as each user's data is heavily perturbed. Hence, $\epsilon$ is also called the privacy budget that controls the trade-off between privacy and utility in DP. With the security guarantee from DP, an external attacker model cannot infer which user's data is used to produce the output $t^{*}$ (e.g., the recommendation results) with high confidence. 

\textbf{Privacy-Preserving Recommender System.} Let $\mathcal{G} = (\mathcal{U} \cup \mathcal{V}, \mathcal{E})$ denote a weighted bipartite graph. $\mathcal{U} = \{u_{1}, u_{2}, ..., u_{|\mathcal{U}|} \}$ and $\mathcal{V} = \{v_{1}, v_{2}, ..., v_{|\mathcal{V}|}\}$ are the sets of users and items. A weighted edge $(u, v, r_{uv}) \in \mathcal{E}$ means that user $u$ has rated item $v$, with weight $r_{uv}$ as $1$. We use $\mathcal{N}(u)$ to denote the set of items rated by $u$ and $\mathcal{N}(v)$ to denote all users who have rated item $v$. Following \cite{beigi2020privacy}, for each user $u$ we construct a dense input vector $\mathbf{x}_u \in \mathbb{R}^{d_0}$ with each element representing either a sensitive attribute $s\in \mathcal{S}$ or a pre-defined statistical feature $s \in \mathcal{S}'$ of $u$. All categorical features are represented by one-hot encodings in $\textbf{x}_u$, while all numerical features are further normalized into $[-1, 1]$. We define the target of a privacy-preserving recommender system below.

\textbf{Problem 1.} Given the weighted graph $\mathcal{G}$ and user feature vectors $\{\textbf{x}_u|u\in \mathcal{U}\}$, we aim to learn a privacy-preserving recommender system that can recommend $K$ products of interest to each user, while any malicious attacker model cannot accurately infer users' sensitive attributes (i.e., gender, occupation and age in our case) from the users' interaction data including both the users' historical ratings and current recommendation results. It is worth noting that our goal is to protect users against a malicious attacker, but not against the recommender system that is trusted.
\begin{figure*}[t]
    \centering
    \scalebox{0.55}{%
    \includegraphics{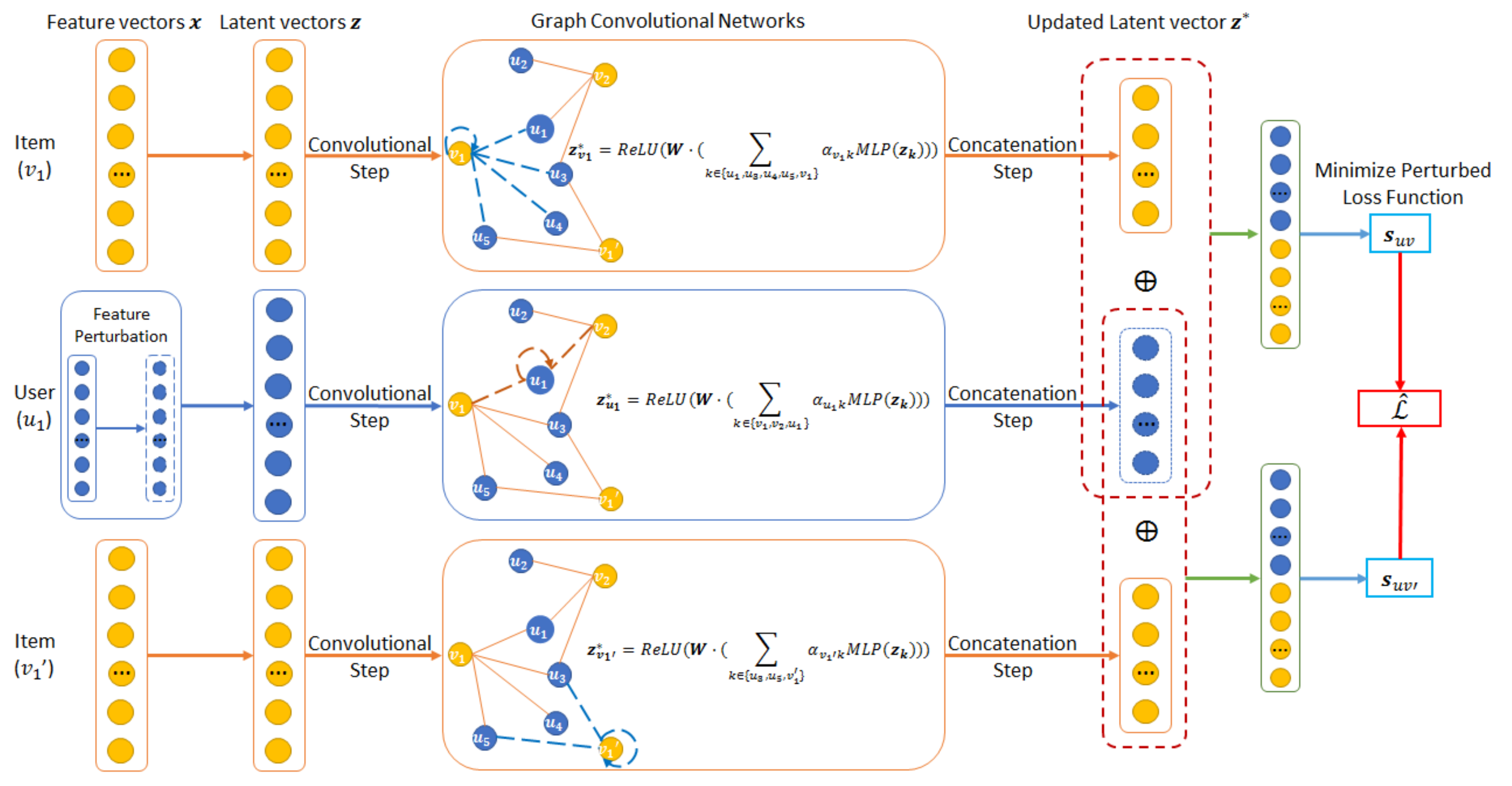}}
    \caption{The overview of GERAI}
   \vspace{-1.5em}
    \label{fig:framework}
\end{figure*}

\section{GCN-based Recommendation Module}\label{st:gcn}
As we aim to address the privacy concerns in GCN-based recommendation models, in this work we build our base recommender upon GCNs \cite{kipf2016semi,hamilton2017inductive}. A recommender, at its core, learns vector representations (a.k.a. embeddings) of both users and items based on their historical interactions, then a user's interest on each item can be easily inferred by measuring the user-item similarity in the latent vector space. When performing recommendation on the graph-structured data, owing to the ability to preserve a graphs topological structure, GCNs can produce highly expressive user and item embeddings for recommendation. Given a 
weighted graph $\mathcal{G} = (\mathcal{U} \cup \mathcal{V}, \mathcal{E})$, users and items are two types of nodes connected by observed links. Then, for each node, GCN computes its embedding by iteratively aggregating information from its local neighbors, where all node embeddings are optimized for predicting the affinity of each user-item pair for personalized ranking.

We first introduce our recommendation module from the user side. For each user $u$, the information $\mathcal{I}(u)$ passed into $u$ comes from the user's first-order neighbors, i.e., items rated by $u$:
\begin{equation}
\begin{split}
	\mathcal{I}(u) &= \{\mathbf{m}_v|v\in \mathcal{N}(u)\} \cup \{\textbf{m}_u\} \\
    &= \{MLP(\mathbf{z}_v)|v\in \mathcal{N}(u)\} \cup\{MLP(\textbf{z}_u)\},
\end{split}
\end{equation}
where $MLP(\cdot)$ is a multi-layer perceptron, $\mathbf{m}_u$/$\mathbf{m}_v$ are the messages from user/item nodes, and $\textbf{z}_u, \textbf{z}_v \in \mathbb{R}^d$ respectively denote the learnable latent embeddings of user $u$ and item $v$. Note that $\textbf{z}_u, \textbf{z}_v$ can be initialized as follows:
\begin{equation}
	\textbf{z}_u = \textbf{E}_{\mathcal{U}}\textbf{x}_u, \,\,\textbf{z}_v = \textbf{E}_{\mathcal{V}}\textbf{x}_v,
\end{equation}
where $\textbf{x}_u \in \mathbb{R}^{d_0}$ is user $u$'s raw feature vector and $\textbf{E}_{\mathcal{U}} \in \mathbb{R}^{d\times d_0}$ is the user embedding matrix. $\textbf{x}_v \in \mathbb{R}^{d_1}$ and $\textbf{E}_{\mathcal{V}} \in \mathbb{R}^{d\times d_1}$ are respectively the item feature vector and embedding matrix. To ensure our model's generalizability, we formulate $\textbf{x}_v$ as an item's one-hot encoding as we do not assume the availability of item features. 
Then, an aggregation operation is performed to merge all information in $\mathcal{I}(u)$, thus forming an updated user embedding $\textbf{z}_u^{*}$: 
\begin{equation}\label{eq:convolution}
    \mathbf{z}_{u}^{*} = ReLU(\mathbf{W} \cdot Aggregate(\mathcal{I}(u))+ \mathbf{b}),
\end{equation}
where $Aggregate(\cdot)$ is the aggregation function and $ReLU(\cdot)$ denotes the rectified linear unit for nonlinearity, and $\mathbf{W}$ and $\mathbf{b}$ are learnable weight matrix and bias vector. 
Motivated by the effectiveness of attention mechanism \cite{velivckovic2017graph} in graph representation learning, we quantify the varied contributions of each element in $\mathcal{I}(u)$ to embedding $\textbf{z}_u^*$ by assigning each neighbour node a different weight. Formally, we define $Aggregate(\mathcal{I}(u))$ as:
\begin{equation}
    Aggregate(\mathcal{I}(u)) = \sum_{k\in \mathcal{N}(u)\cup \{u\}}\alpha_{uk}\mathbf{m}_{k},
\end{equation}
where $\alpha_{uk}$ denotes the attention weight implying the importance of message $\mathbf{m}_k\in\mathcal{I}(u)$ to user node $u$ during aggregation. Specifically, to compute $\alpha_{uk}$, we first calculate an attention score $a_{uk}$ via the following attention network: 
\begin{equation}
    a_{uk} = {\mathbf{w}_{2}^{\top}} \cdot \sigma(\mathbf{W}_{1} (\mathbf{m}_{k} \oplus \mathbf{z}_{u}) + \mathbf{b}_{1}) + b_{2},
\end{equation}
where $\oplus$ represents the concatenation of two vectors. Afterwards, each final attention weight $\alpha_{uk}$ is computed by normalizing all the attentive scores using softmax:
\begin{equation}
      \alpha_{uk} = \frac{\exp(a_{uk})}{\sum_{k'\in \mathcal{N}(u)\cup\{u\}} \exp(a_{uk'})}.
\end{equation}

Likewise, on the item side, we repeat the message passing scheme by aggregating the information from an item's interacted users in $\mathcal{N}(v)$ to learn the item embedding $\textbf{z}_v^*$:
\begin{equation}
    \mathbf{z}_{v}^{*} = ReLU(\mathbf{W} \cdot Aggregate(\mathcal{I}(v))+ \mathbf{b}),
\end{equation}
where $\mathcal{I}(v) = \{MLP(\textbf{z}_u)|u\in \mathcal{N}(v)\}\cup\{MLP(\textbf{z}_v)\}$. Note that the same network structure and trainable parameters are shared in the computation of both user and item embeddings.

To train our model for top-$K$ recommendation, we leverage the pairwise Bayesian personalized ranking (BPR) loss~\cite{rendle2012bpr} to learn model parameters. To facilitate personalized ranking, we firstly generate a ranking score $s_{uv}$ for an arbitrary user-item tuple $(u,v)$:
\begin{equation}\label{eq:predict}
\begin{split}
    \textbf{q}_{uv} &= ReLU (\mathbf{W}_{3}(\mathbf{z}_u^{*} \oplus \mathbf{z}_v^{*}) + \mathbf{b}_3),\\
    s_{uv} &= \mathbf{h}^{\top}\textbf{q}_{uv},
\end{split}
\end{equation}
where $\textbf{h}\in \mathbb{R}^d$ is the projection weight. Intuitively, BPR optimizes ranking performance by comparing two ranking scores $s_{uv}$, $s_{uv'}$ for user $u$ on items $v$ and $v'$. In each training case $(u, v, v')$, $v$ is the positive item sampled from $\mathcal{E}$, while $v'$ is the negative item having $r_{uv'} \notin \mathcal{E}$. Then, BPR encourages that $v$ should have a higher ranking score than $v'$ by enforcing:
\begin{equation}\label{eq:rec_loss}
    \mathcal{L} = \sum_{(u,v,v')\in \mathcal{D}}-\log \sigma(s_{uv}-s_{uv'}) +\gamma ||\Theta||^2,
\end{equation}
where $\mathcal{D}$ is the training set, $\sigma(\cdot)$ is the sigmoid function, $\Theta$ denotes parameters in the GCN-based recommendation module, and $\gamma$ is the L2-regularization coefficient.

\section{GERAI: Graph Embedding for Recommendation against Attribute Inference}
In this section, we formally present the design of GERAI, a recommendation model that can defend attribute inference attacks via a novel dual-stage differential privacy constraint. Figure \ref{fig:framework} depicts the workflow of GERAI, where two important perturbation operations take place at both the input stage for user features and the optimization stage for the loss function.
The first step is to achieve $\epsilon^{\rhd}-$local differential privacy ($\epsilon^{\rhd}-$LDP) by directly adding noise to users' raw feature vectors $\textbf{x}_u$ used for learning user embeddings, which can avoid exposing users' sensitive data to an unsecured cyber environment during upload, while providing the GCN-based recommender with side information for learning expressive user representations. Then, to prevent GERAI from generating recommendation results that can reveal users' sensitive attributes, we further enforce $\epsilon-$DP in the optimization stage by perturbing its loss function $\mathcal{L}$. However, this is a non-trivial task as it requires to calculate the privacy sensitivity of $\mathcal{L}$, which involves analyzing the complex relationship between the input data and learnable parameters. Hence, we propose a novel solution by deriving a polynomial approximation $\widetilde{\mathcal{L}}$ of the original BPR loss $\mathcal{L}$, so as to support sensitivity calculation and perform perturbation on $\widetilde{\mathcal{L}}$ to facilitate differentially private training of GERAI. Notably, to distinguish the DP constraints in two stages, we denote $\epsilon^{\rhd}$ as local privacy budget and $\epsilon$ as global privacy budget, respectively.
\vspace{-1em}
\subsection{User Feature Perturbation at Input Stage}
\begin{algorithm}[t!]
            \caption{Perturbing 1-Dimensional Numerical Data with Piecewise Mechanism}
            \KwIn{A single numerical feature $x \in [-1,1]$ and coefficient $\epsilon^{\rhd}$}
            \KwOut{Perturbed feature $\widehat{x} \in [-C, C]$}
             Uniformly sample $\xi$ from $[0,1]$\;
          \eIf{$\xi < \frac{\exp(\frac{\epsilon^{\rhd}}{2})}{\exp(\frac{\epsilon^{\rhd}}{2})+1}$}
               {Uniformly sample $\widehat{x}$ from $[\ell(x), \pi(x)]$\;}
               {Uniformly sample $\widehat{x}$ from $[-C, \ell(x)\big{)} \cup \big{(}\pi(x), C]$\;}
           \textbf{return} $\widehat{x}$
\label{alg:perturbation}          
\end{algorithm}
\setlength{\textfloatsep}{1em}
At the input level, the feature vector $\textbf{x}_{u}$ of each user $u$ is perturbed before being fed into the recommender module. This helps address users' privacy concerns on sharing their personal attributes and keep them confidential during the upload process. Furthermore, as we will show in Section \ref{sec:ablation}, perturbing user features contributes to defending attribute inference attacks as the recommendation results are no longer based on the actual attributes.
Then, instead of the original $\textbf{x}_u$, the perturbed data $\widehat{\textbf{x}}_u$ will be used for the recommendation purpose. To achieve this, we treat numerical and categorical features separately, as these two types of data will require different perturbation strategies. Firstly, for numerical data, perturbation is performed based on a randomized encryption mechanism named piecewise mechanism (PM) \cite{wang2019collecting}. Algorithm~\ref{alg:perturbation} shows the PM-based perturbation for each scalar numerical feature $x\in \textbf{x}_u$. In PM, the original feature $x \in [-1, 1]$ will be transformed into a perturbed value $\widehat{x} \in [-C, C]$, with $C$ defined as follows: 
\begin{equation}
    C = \frac{\exp(\frac{\epsilon^{\rhd}}{2})+1}{\exp(\frac{\epsilon^{\rhd}}{2})-1}.
\end{equation}

The probability density function of the noisy output $\widehat{x}$ is:
\begin{equation}\label{eq:prob_density}
    Pr(\widehat{x} = c|x) = \begin{cases}
 p,& \text{ if } c\in[\ell(x),\pi(x)] \\ 
 \frac{p}{\exp(\epsilon^{\rhd})},& \text{ if } c \in [-C, \ell(x)\big{)} \cup \big{(}\pi(x), C]
\end{cases},
\end{equation}
where:
\begin{equation}
\begin{split}
p&=\frac{exp(\epsilon^{\rhd})-exp(\epsilon^{\rhd}/2)}{2exp(\epsilon^{\rhd}/2)+2},\\
\ell(x) &= \frac{C+1}{2} \cdot x - \frac{C-1}{2}, \\
\pi(x) &= \ell(x) + C - 1.
\end{split}
\end{equation}

The following lemma establishes the theoretical guarantee of Algorithm~\ref{alg:perturbation}.
\begin{lemma}\label{lemma:1}
Algorithm~\ref{alg:perturbation} satisfies $\epsilon^{\rhd}-$local differential privacy.
\end{lemma}
\begin{proof}
By Eq.(\ref{eq:prob_density}), let $x, x' \in [-1, 1]$ be any two input values and $\widehat{x} \in [-C, C]$ denote the output of Algorithm~\ref{alg:perturbation}, then we have:
\begin{equation}
\frac{Pr(\widehat{x}|x)}{Pr(\widehat{x}|x')} \leq \frac{p}{p/\exp(\epsilon^{\rhd})} = \exp(\epsilon^{\rhd}).
\end{equation}
Thus, Algorithm 1 satisfies $\epsilon^{\rhd}-$LDP.
\end{proof}
\begin{algorithm}[t!]
            \caption{Perturbing Multidimensional Data with Numerical and Categorical Features}
            \KwIn{Feature vector $\mathbf{x}_u = \textbf{x}_{(1)}\oplus\textbf{x}_{(2)}\oplus\cdots\oplus\textbf{x}_{(d')} \in \mathbb{R}^{d_0}$ and coefficient $\epsilon^{\rhd}$}
            \KwOut{Perturbed feature vector $\widehat{\mathbf{x}}_u$} 
            $\widehat{\mathbf{x}}_u = \widehat{\textbf{x}}_{(1)}\oplus\widehat{\textbf{x}}_{(2)}\oplus\cdots\oplus\widehat{\textbf{x}}_{(d')} \leftarrow \{0\}^{d_0}$\; 
            $\mathcal{A} \leftarrow \zeta$ different values uniformly sampled from $\{1,2,...,d'\}$\;
            \For{each feature index $i \in \mathcal{A}$}{
                \eIf{$\textbf{x}_{(i)}$ is a numerical feature}{
                  $\widehat{\textbf{x}}_{(i)} \!\leftarrow$ Execute Algorithm ~\ref{alg:perturbation} with $x\!=\textbf{x}_i$ and $\epsilon^{\rhd} \!=\! \frac{\epsilon^{\rhd}}{\zeta}$\;
                  $\widehat{\textbf{x}}_{(i)} \!\leftarrow \frac{d'}{\zeta}\widehat{\textbf{x}}_{(i)}$\;
                   }
                   {fetch categorical feature $i$'s one-hot encoding from $\textbf{x}_u$, denoted by $\textbf{c}\leftarrow \mathbf{x}_{(i)} = [0,\cdots,0,1,0,\cdots,0]$\;
                    \For{each element $c \in \mathbf{c}$}
                    {
                    Draw $c'$ from $\{0,1\}$ with $Pr[c'=1] = \begin{cases}
                     0.5,& \text{ if } c' = 1 \\ 
                     \frac{1}{\exp(\frac{\epsilon^{\rhd}}{\zeta})+1},& \text{ if } c = 0. \end{cases}$\;
                    $c\leftarrow c'$\;  
                    }
                    $\widehat{\textbf{x}}_{(i)} \leftarrow \textbf{c}$\;
                  }
            }
           return $\widehat{\textbf{x}}_u$
\label{alg:perturbation2}
\end{algorithm}
However, the PM perturbation presented above is only designed for numerical data that is 1-dimensional. Hence, inspired by~\cite{wang2019collecting}, we generalize Algorithm~\ref{alg:perturbation} to the multidimensional $\textbf{x}_u$ containing both numerical and categorical attributes. Given $\textbf{x}_u\in\mathbb{R}^{d_0}$, considering it encodes $d'$ different features in total, we can rewrite it as $\textbf{x}_u = \textbf{x}_{(1)}\oplus\textbf{x}_{(2)}\oplus\cdots\oplus\textbf{x}_{(d')}$, where the $i$-th feature $\textbf{x}_{(i)}$ ($1\leq i \leq d'$) is either an one-dimensional numeric or an one-hot encoding vector for a categorical feature. On this basis, we propose a comprehensive approach for perturbing such multidimensional data. The detailed perturbation process is depicted in Algorithm~\ref{alg:perturbation2}. Noticeably, we only perturb $\zeta<d'$ features in $\textbf{x}_u$. This is because that, if we straightforwardly treat each of the $d'$ features in $\textbf{x}_u$ as an individual element in the dataset, then according to the composition theorem~\cite{dwork2014algorithmic}, the local privacy budget for each feature will shrink to $\frac{\epsilon^{\rhd}}{d'}$ in order to maintain $\epsilon^{\rhd}-$LDP. As a consequence, this will significantly harm the utility of encrypted data. Hence, to preserve reasonable quality of each perturbed numerical or categorical feature, we propose to encrypt only a fraction of (i.e., $\zeta$) features in $\textbf{x}_u$, ensuring a higher local privacy budget of $\frac{\epsilon^{\rhd}}{\zeta}$. As shown in Algorithm~\ref{alg:perturbation2}, to prevent privacy leakage, the unselected $d'-\zeta$ features will be dropped by masking them with $0$. Thus, to offset the recommendation accuracy loss caused by dropping these features, we follow the empirical study in \cite{wang2019collecting} to determine the appropriate value of $\zeta$:
\begin{equation}
     \zeta = \max \{1, \min\{d', \lfloor\frac{\epsilon^{\rhd}}{2.5} \rfloor\}\}.
\end{equation}
Additionally, when perturbing each categorical feature $\textbf{x}_{(i)}\in\textbf{x}_u$, we extend the continuous sampling strategy in Algorithm~\ref{alg:perturbation} to a binarized version for each element/bit within the one-hot encoding $\textbf{x}_{(i)}$ with the updated local privacy budget $\frac{\epsilon^{\rhd}}{\zeta}$. As the privacy guarantee of the perturbed categorical feature $\widehat{\textbf{x}}_{(i)}$ can be verified in a similar way to numerical features~\cite{wang2017locally}, we have omitted this part to be succinct. In this regard, our perturbation strategy for the user-centric data in recommendation can provide $\epsilon^{\rhd}-$LDP, as we summarize below:
\begin{lemma}
Algorithm 2 satisfies $\epsilon^{\rhd}-$local differential privacy.
\end{lemma}
\begin{proof}
As Algorithm~\ref{alg:perturbation2} is composed of $\zeta$ times of $\frac{\epsilon^{\rhd}}{\zeta}-$LDP operations, then based on the composition theorem~\cite{dwork2014algorithmic}, Algorithm~\ref{alg:perturbation2} satisfies $\epsilon^{\rhd}-$LDP.
\end{proof}

\subsection{Loss Perturbation at Optimization Stage}
\begin{algorithm}[t]
\caption{Optimizing GERAI}
\KwIn{Maximum iteration number $\mathcal{T}$, coefficient $\epsilon$ and learning rate $\eta$}\
\KwOut{Optimal Parameters ${\Theta}^*$ of GERAI}\
$\Delta \leftarrow d + \frac{d^2}{4}$ \;
\For{$0\leq j\leq 2$}
{
\For{$\phi \in \Phi_{j}$}{
$\lambda_{\phi} \leftarrow \sum_{t\in \mathcal{D}}\lambda_{\phi t} + Lap(\frac{\Delta}{\epsilon |\mathcal{D}|})$, where $t=(u, v, v')$ denotes a triplet training sample\;
}
}
$\widehat{\mathcal{L}} \leftarrow \sum_{j=0}^{2}\sum_{t\in\mathcal{D}}\lambda_{\phi t}(\textbf{h}^{\top}\textbf{q}_{uv}-\textbf{h}^{\top}\textbf{q}_{uv'})$, where $\widehat{\mathcal{L}}$ is the perturbed loss\;
Initialize ${\Theta}^*$ randomly\;
\For{each $u\in \mathcal{U}$}
{$\widehat{\textbf{x}}_u \leftarrow$ Algorithm~\ref{alg:perturbation2}\;}
\For{$t \in \mathcal{T}$}
{
Draw a minibatch $\mathcal{B}$ \;
$\widehat{\mathcal{L}} \leftarrow$ Eq.(\ref{eq:loss_approx})\;
Take a gradient step to optimize ${\Theta}^*$ with learning rate $\eta$\;
}
\textbf{Return} ${\Theta}^*$.
\label{alg:loss_perturbation}
\end{algorithm}

In most scenarios, the results generated by a predictive model (e.g., models for predicting personal credit or diseases) carry highly sensitive information about a user, and this is also the case for recommender systems, since the recommended items can be highly indicative on a user's personal interests and demographics. Though privacy can be achieved via direct perturbation on the generated results \cite{lei2011differentially,chaudhuri2011differentially}, it inevitably impedes a model's capability of learning an accurate mapping from its input to output \cite{zhang2012functional}, making the learned recommender unable to fully capture personalized user preferences for recommendation. Hence, in the recommendation context, we innovatively propose to perturb the ranking loss $\mathcal{L}$ (i.e., Eq.(\ref{eq:rec_loss})) instead of perturbing the recommendation results in GERAI. This incurs the analysis of the privacy sensitivity  $\Delta$ of $\mathcal{L}$. For any function, the privacy sensitivity is the maximum L1 distance between its output values 
given two neighbor datasets differing in one data instance. Intuitively, the larger that $\Delta$ is, the heavier perturbation noise is needed to maintain a certain level of privacy. However, directly computing $\Delta$ from $\mathcal{L}$ is non-trivial due to its unbounded output range and the complex association between the input data and trainable parameters.

Hence, we present a novel solution to preserving global $\epsilon-$DP for our ranking task. Motivated by the functional mechanism (FM)~\cite{zhang2012functional} used for loss perturbation in regression tasks, we first derive a polynomial approximation $\widetilde{\mathcal{L}}$ for $\mathcal{L}$
to allow for convenient privacy sensitivity computation and make the private-preserving optimization process more generic. Then, GERAI perturbs $\widetilde{\mathcal{L}}$ by injecting Laplace noise to enforce $\epsilon-$DP. It is worth noting that, to calculate the privacy sensitivity of $\widetilde{\mathcal{L}}$, we apply a normalization step\footnote{This assumption can be easily enforced by the clip function.} to every latent predictive feature $\mathbf{q}_{uv}$ produced in Eq.(\ref{eq:predict}), which ensures every element in $\mathbf{q}_{uv}$ is bounded by $(0,1)$. Using Taylor expansion, we derive $\widetilde{\mathcal{L}}$, the polynomial approximation of $\mathcal{L}$:
\begin{equation}\label{eq:loss_approx}
\resizebox{.86\linewidth}{!}{
    $\widetilde{\mathcal{L}} = \frac{1}{|\mathcal{D}|}\sum_{\forall (u,v,v') \in \mathcal{D}}\sum_{j = 0}^{\infty} \frac{f^{(k)}(0)}{k!}(\textbf{h}^{\top}\textbf{q}_{uv}-\textbf{h}^{\top}\textbf{q}_{uv'})^{j}
$}
\end{equation}
where $\frac{f^{(k)}(0)}{k!}$ is the k-th derivative of $\widetilde{\mathcal{L}}$ at $0$. Recall that $\mathbf{h} = [h_1,h_2,...,h_d]$ is a projection vector containing $d$ values. Let $\phi(\textbf{h}) = h^{c_1}_1 h^{c_2}_2 \cdots h^{c_d}_d$ for $c_1,...,c_d \in \mathbb{N}$. Let $\Phi_j = \{h^{c_1}_1 h^{c_2}_2 \cdots h^{c_d}_d|\sum_{l=1}^{d}c_l = j\}$ given the degree $j$ (e.g., $\Phi_0= \{1\}$). Following~\cite{zhang2012functional}, we truncate the Taylor series in $\widetilde{\mathcal{L}}$ to retain polynomial terms with order lower than $3$. Specially, only $\Phi_0, \Phi_1$ and $\Phi_2$ involved in $\widetilde{\mathcal{L}}$ with polynomial coefficients as $\frac{f^{(0)}(0)}{0!} = log^{2}, \frac{f^{(1)}(0)}{1!} = -\frac{1}{
2}, \frac{f^{(2)}(0)}{2!} = \frac{1}{8}$. 

Based on $\widetilde{\mathcal{L}}$, we now explore the global privacy sensitivity of the recommendation loss, denoted as $\Delta$. Let $\lambda_{\phi t} \in \mathbb{R}$ denote the coefficient of $\phi(\mathbf{h})$ in the polynomial. In each mini-batch training iteration, the difference of input data only influences these coefficients, so we add perturbation to $\widetilde{\mathcal{L}}$'s coefficients based on the sensitivity. In the following lemma, we derive the global sensitivity $\Delta$ of $\widetilde{\mathcal{L}}$, which serves as the important scale factor in determining the noise intensity: 
\begin{lemma}
The global sensitivity of $\widetilde{\mathcal{L}}$ is $d + \frac{d^2}{4}$.
\end{lemma}
\begin{proof} Given $\widetilde{\mathcal{L}}$ and two training datasets $\mathcal{D}$, $\mathcal{D}'$ that differ in only one instance, for $J\geq 1$ and $\overline{\textbf{q}}= [\overline{q}_1,\overline{q}_2,...,\overline{q}_d]= \mathbf{q}_{uv}-\mathbf{q}_{uv'}$, we can derive:
\begin{small}
\begin{equation}\label{eq:Delta}
\begin{split}
    \Delta &= \sum_{j=1}^{J}\sum_{\phi \in \Phi_j}||\sum_{t \in \mathcal{D}}\lambda_{\phi t} - \sum_{t' \in \mathcal{D}'}\lambda_{\phi t'}||_1 \\
    &\leq 2 \cdot \underset{t}{\max} \sum_{j = 1}^{J}\sum_{\phi \in \Phi_{j}}||\lambda_{\phi t}||_{1}\\
    &\leq 2 \cdot \underset{t}{\max} \Big{(}\frac{f^{(1)}(0)}{1!}\sum_{m=1}^{d}\overline{q}_{m}\Big{)} +\frac{f^{(2)}(0)}{2!}\sum_{m\geq1,n\leq d}\overline{q}_{m}\overline{q}_{n} \\
    &\leq 2(\frac{\textnormal{dim}(\mathbf{q}_{uv})}{2}+\frac{\textnormal{dim}(\mathbf{q}_{uv})^2}{8}) \\
    &= d + \frac{d^2}{4},
\end{split}
\end{equation}
\end{small}
where $t = (u, v, v') \in \mathcal{D}$ is an arbitrary training sample and $\textnormal{dim}(\cdot)$ returns the dimension of a given vector.
\end{proof}
Specifically, we employ FM to perturb the loss $\widetilde{\mathcal{L}}$ by injecting Laplace noise\footnote{In our paper, the mean of our Laplace distribution is 0, i.e., $Lap(\cdot)= Lap(0,\cdot)$.} $Lap(\frac{\Delta}{\epsilon |\mathcal{D}|})$ into its polynomial coefficients, and the perturbed function is denoted by $\widehat{\mathcal{L}}$. The injected Laplace noise with standard deviation of $\frac{\Delta}{\epsilon |\mathcal{D}|}$ has been widely proven to effectively retain $\epsilon-$DP after perturbation \cite{dwork2014algorithmic,dwork2006calibrating,zhang2012functional}. Note that as $\Delta$ is the global sensitivity, it is evenly distributed to all instances in the training set $\mathcal{D}$ during perturbation. We showcase the full training process of GERAI with a differentially private loss in Algorithm~\ref{alg:loss_perturbation}. In Algorithm~\ref{alg:loss_perturbation}, we first compute the sensitivity $\Delta$ of loss $\widetilde{\mathcal{L}}$. In each iteration, we add perturbation to every coefficient in the polynomial approximation of the loss function. Afterwards, we launch the training session for GERAI with perturbed user feature vectors $\{\widehat{\textbf{x}}_u|u\in\mathcal{U}\}$, where we use the perturbed coefficients to obtain the perturbed loss $\widehat{\mathcal{L}}$ and optimize the parameters of the model by minimizing $\widehat{\mathcal{L}}$. Finally, we 
formally prove that Algorithm~\ref{alg:loss_perturbation} satisfies $\epsilon-$DP:
\begin{lemma}
Algorithm 3 maintains $\epsilon-$differential privacy.
\end{lemma}
\begin{proof}Assume that $\mathcal{D}$ and $\mathcal{D}^{'}$ are two training datasets differing in only one instance denoted by $T$ and $T'$, then we have:
\begin{small}
\begin{equation}\normalsize
\begin{split}
    \frac{Pr(\widehat{\mathcal{L}}|{\mathcal{D}})}{Pr(\widehat{\mathcal{L}}|\mathcal{D}')} &= \frac{\Pi_{j=1}^{2}\Pi_{\phi \in \Phi_{j}}\exp(\frac{\epsilon}{\Delta} ||\sum_{t \in \mathcal{D}}\lambda_{\phi t}-\lambda_{\phi}||_1)}{\Pi_{j=1}^{2}\Pi_{\phi \in \Phi_{j}}\exp(\frac{\epsilon}{\Delta} ||\sum_{t' \in \mathcal{D}'}\lambda_{\phi t'}-\lambda_{\phi}||_1)}\\
    &\leq \Pi_{j=1}^{2} \Pi_{\phi \in \Phi_{j}} \exp(\frac{\epsilon}{\Delta} ||\sum_{t\in \mathcal{D}}\lambda_{\phi t}-\sum_{t'\in \mathcal{D}'}\lambda_{\phi t'}||_{1})\\
    &=\Pi_{j=1}^{2} \Pi_{\phi \in \Phi_{j}} \exp(\frac{\epsilon}{\Delta}||\lambda_{\phi T}-\lambda_{\phi T'}||_{1})\\
    &= \exp(\frac{\epsilon}{\Delta}\sum_{j=1}^{2}\sum_{\phi\in\Phi_{j}}||\lambda_{\phi T} -\lambda_{\phi T'}||_{1})\\
    &\leq \exp(\frac{\epsilon}{\Delta} \cdot 2 \cdot \underset{T}{\max}\sum_{j=1}^{2}\sum_{\phi\in\Phi_{j}}||\lambda_{\phi T}||_{1}) = \exp(\epsilon).
\end{split}
\end{equation}
\end{small}
Then according to Definition 1, Algorithm~\ref{alg:loss_perturbation} satisfies $\epsilon-$DP.
\end{proof}
In short, with our proposed dual-stage perturbation strategy for both the user data and the training loss, GERAI fully preserves user privacy with a demonstrable guarantee, while being able to achieve minimal compromise on the recommendation effectiveness compared with a non-private, GCN-based counterpart. Furthermore, GERAI can be trained via stochastic gradient descent (SGD) algorithms in an end-to-end fashion, showing its real-world practicality.

\section{Experiments}
\begin{table}[t!]
\caption{Features extracted from the dataset.}
  \scalebox{0.9}{%
  \begin{tabular}{p{9cm}}
    \toprule
    \midrule
\textbf{- Number of rated products}\\
\textbf{- Number and ratio of each rating level given by a user}\\
\textbf{- Ratio of positive and negative ratings}: The proportions of high ratings (4 and 5) and low ratings (1 and 2) of a user.\\
\textbf{- Entropy of ratings}: It is calculated as $-\sum_{\forall r}Prop_{r}\log Prop_{r}$, where $Prop_{r}$ is the proportion that a user gives the rating of $r$.\\
\textbf{- Median, min, max, and average of ratings}\\
\textbf{- Gender}: It is either male or female.\\
\textbf{- Occupation}: A total of 21 possible occupations are extracted.\\
\textbf{- Age}: We categorize age attribute into 3 groups: over 45, under 35, and between 35 and 45.\\
    \midrule
   \bottomrule
\end{tabular}}
\label{tab:ff}
\end{table}
In this section, we conduct experiments to evaluate the performance of GERAI in terms of both privacy strength and recommendation effectiveness. Particularly, we aim to answer the following research questions (RQs):  
\begin{itemize}[leftmargin=*]
    \item \textbf{RQ1:} Can GERAI effectively protect sensitive user data from attribute inference attack?
    \item \textbf{RQ2:} How does GERAI perform in top-$K$ recommendation?
    \item \textbf{RQ3:} How does the key hyperparameters affect the privacy-preserving property and recommendation accuracy of GERAI?
    \item \textbf{RQ:4} What is the contribution from each part of the dual-stage perturbation paradigm in GERAI? 
    \item \textbf{RQ5:} Can GERAI defend different types of unseen attribute inference attack models?
\end{itemize}
\subsection{Dataset}
Following \cite{beigi2020privacy}, we use the publicly available ML-100K datasets~\cite{grouplens_2019} in our experiments. It contains $10,000$ ratings from $943$ users on $1,682$ movies collected from the MovieLens website. In addition, in the collected dataset, each user is associated with three sensitive attributes, i.e., gender (Gen), age (Age) and occupation (Occ). Similar to~\cite{beigi2020privacy}, we convert the gender, age and occupation into a $2$, $3$ and $21$-dimensional categorical feature, respectively. Table \ref{tab:ff} provides a summary of all the features we have used. 
\subsection{Baseline Methods and Parameter Settings}
We evaluate GERAI by comparing with the following baselines:
\begin{itemize}[leftmargin=*]
    \item \textbf{BPR}: It is a widely used non-private learning-to-rank model for recommendation~\cite{rendle2012bpr}.
	\item \textbf{GCN}: This is the non-private, GCN-based recommendation model proposed in \cite{ying2018graph}. 
	\item \textbf{Blurm}: This method directly uses perturbed user-item ratings to train the recommender system~\cite{weinsberg2012blurme}.
	\item \textbf{DPAE}: In DPAE, Gaussian mechanism is combined in the stochastic gradient descent process of an autoencoder-based recommender so that the training phase meets the requirements of differential privacy~\cite{liu2019differentially}.
	\item \textbf{DPNE}: It aims to develop a differentially private network embedding method based on matrix factorization, and it is the state-of-the-art privacy preserving network embedding method for link prediction~\cite{xu2018dpne}.
	\item \textbf{DPMF}: It uses objective perturbation with matrix factorization to ensure the final item profiles satisfy differential privacy~\cite{jingyu1763differentially}.
	\item \textbf{RAP}: It is the state-of-the-art recommendation model that is designed against attribute inference attacks~\cite{beigi2020privacy}. The key idea is to facilitate adversarial learning with an RNN-based private attribute inference attacker and a CF-based recommender.
\end{itemize}
In GERAI, we set $\gamma$, learning rate and batch size to $0.01$, $0.005$ and $64$, respectively. Without special mention, we use three-layer networks for the neural components and initialized parameters to random values by using Gaussian distribution, which has $0$ mean and a standard deviation of $1$. The final embedding dimension is $d = 60$ and the privacy budget is $\epsilon = 0.4$ and $\epsilon^{\rhd} = 20$, while the effect of different hyperparameter values will be further discussed in Section~\ref{sec:hyper_analysis}. For all baseline methods, we use the optimal hyperparameters provided in the original papers. 
\vspace{-0.5em}
\subsection{Evaluation Protocols}\label{sec:protocol}
\begin{table}[t!]
\vspace{-0.5em}
 \caption{Attribute inference attack results. Lower F1 scores represent better privacy protection from the model.}
\centering
\scalebox{0.9}{%
 \begin{tabular}{|p{1.2cm}<{\centering} | p{1.0cm}<{\centering}| p{0.7cm}<{\centering} |p{0.7cm}<{\centering} |p{0.7cm}<{\centering}| p{0.7cm}<{\centering}|p{0.7cm}<{\centering}|p{0.7cm}<{\centering}|} 
 \hline
 \multirow{2}{*}{Attribute}&\multirow{2}{*}{Method}&\multicolumn{6}{c|}{F1 Score}\\
  \cline{3-8}
  &&K=5&K=10&K=15&K=20&K=25&K=30\\
 \hline
 \multirow{7}{*}{Age}&BPR&0.693&0.694&0.699&0.720&0.676&0.693\\
 &GCN&0.697&0.725&0.730&0.725&0.735&0.746\\
 &Blurm&0.715&0.725&0.716&0.692&0.679&0.710\\
 &DPAE&0.694&0.688&0.695&0.674&0.695&0.684\\
 &DPNE&0.684&0.685&0.700&0.701&0.679&0.674\\
 &DPMF&0.709&0.703&0.695&0.699&0.684&0.689\\
 &RAP&\textbf{0.661}&\textbf{0.650}&0.677&0.666&0.674&0.671\\
 \cline{2-8}
 &GERAI&0.677&0.663&\textbf{0.648}&\textbf{0.651}&\textbf{0.652}&\textbf{0.650}\\
 \hline
 \hline
 \multirow{7}{*}{Gen}
 &BPR&0.810&0.773&0.808&0.778&0.782&0.801\\
 &GCN&0.851&0.836&0.891&0.880&0.862&0.869\\
 &Blurm&0.789&0.788&0.789&0.761&0.761&0.788\\
 &DPAE&0.781&0.771&0.770&0.772&0.771&0.777\\
 &DPNE&0.788&0.772&0.781&0.776&0.798&0.788\\
 &DPMF&0.783&0.770&0.768&0.765&0.761&0.771\\
 &RAP&0.787&0.771&0.763&0.772&0.776&0.763\\
 \cline{2-8}
 &GERAI&\textbf{0.760}&\textbf{0.755}&\textbf{0.763}&\textbf{0.760}&\textbf{0.744}&\textbf{0.755}\\
 \hline
 \hline
 \multirow{7}{*}{Occ}
 &BPR&0.276&0.277&0.264&0.263&0.289&0.267\\
 &GCN&0.277&0.277&0.277&0.267&0.272&0.270\\
 &Blurm&0.267&0.267&0.262&0.262&0.267&0.269\\
 &DPAE&0.266&0.260&0.255&0.261&0.260&0.261\\
 &DPNE&0.267&0.265&0.266&0.264&0.266&0.262\\
 &DPMF&0.266&0.262&0.270&0.265&0.270&0.267\\
 &RAP&0.260&0.262&0.260&0.263&0.248&0.260\\
 \cline{2-8}
 &GERAI&\textbf{0.260}&\textbf{0.261}&\textbf{0.255}&\textbf{0.256}&\textbf{0.246}&\textbf{0.251}\\
 \hline
 \end{tabular}}
 \label{table:attack}
\end{table}
\textbf{Attribute Inference Attack Resistance.} To evaluate all models' robustness against attribute inference attacks, we first build a strong adversary classifier (i.e., attacker). Specifically, we use a two-layer deep neural network model as the attacker. Suppose there are $K$ items $\mathcal{R}(u)$ recommended by a fully trained recommender to user $u\in\mathcal{U}$, then the input of the attacker is formulated as $\sum_{\forall{v\in\mathcal{I}(u)}}{onehot(v)}+\sum_{\forall{v\in\mathcal{R}(u)}}{onehot(v)}$ where $onehot(\cdot)$ returns the one-hot encoding of a given item. The hidden dimension is set to $100$, and a linear projection is used to estimate the class of the target attribute. We randomly choose $80\%$ of the labelled users to train the attacker, and use the remainder to test the attacker's inference accuracy. Note that the attacker model is unknown  to all recommenders during the training process. To quantify a model's privacy-preserving capability, we leverage a widely-used classification metric \textit{F1 score} \cite{ziegler2005improving} to evaluate the classification performance of the attacker. Correspondingly, lower F1 scores demonstrate higher resistance to this inference attack.

\textbf{Recommendation Effectiveness.} For each user, we randomly pick $80\%$ of her/his interacted items to train all recommendation models, while the rest $20\%$ is held out for evaluation. We employ $Hit@K$ and $NDCG@K$, which are two popular metrics to judge the quality of the top-$K$ ranking list. Results on both attribute inference and recommendation are averaged over five runs.
\vspace{-0.5em}
\subsection{Privacy Protection Effectiveness (RQ1)}
\begin{table}[t!]
\vspace{-0.5em}
\caption{Recommendation effectiveness results. For both Hit@K and NDCG@K, the higher the better.}
\centering
\scalebox{0.88}{%
 \begin{tabular}{|p{1.3cm}<{\centering} |p{1.2cm}<{\centering} |p{0.7cm}<{\centering} |p{0.7cm}<{\centering}| p{0.7cm}<{\centering}| p{0.7cm}<{\centering} |p{0.7cm}<{\centering}|p{0.7cm}<{\centering}|} 
 \hline
 &Method&K=5&K=10&K=15&K=20&K=25&K=30\\
 \hline 
 \multirow{7}{*}{Hit@K}&BPR&0.348&0.507&0.614&0.686&0.741&0.791\\
 &GCN&0.365&0.519&0.619&0.690&0.743&0.789\\
 &Blurm&0.184&0.263&0.319&0.364&0.405&0.443\\
 &DPAE&0.185&0.285&0.345&0.394&0.438&0.458\\
 &DPNE&0.301&0.430&0.525&0.595&0.640&0.684\\
 &DPMF&0.195&0.280&0.343&0.394&0.432&0.474\\
 &RAP&0.319&0.475&0.575&0.648&0.706&0.754\\
 \cline{2-8}
 &GERAI&0.333&0.495&0.600&0.670&0.724&0.767\\
 \hline
 \hline 
 \multirow{7}{*}{NDCG@K}&BPR&0.228&0.280&0.310&0.330&0.341&0.363\\
 &GCN&0.247&0.296&0.323&0.340&0.351&0.360\\
 &Blurm&0.124&0.148&0.164&0.174&0.183&0.191\\
 &DPAE&0.126&0.153&0.170&0.176&0.180&0.188\\
 &DPNE&0.204&0.231&0.268&0.289&0.299&0.306\\
 &DPMF&0.134&0.154&0.171&0.182&0.191&0.186\\
 &RAP&0.211&0.264&0.286&0.308&0.317&0.329\\
 \cline{2-8}
 &GERAI&0.217&0.270&0.296&0.314&0.326&0.334\\
 \hline
 \end{tabular}}
 \label{table:rec}
\end{table}
\begin{figure*}[t!]
\centering
\begin{tabular}{cccc}
	\multicolumn{4}{c}{\includegraphics[scale=0.65]{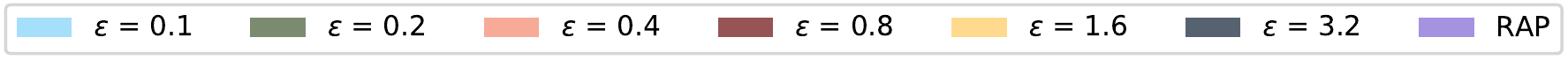}}\\
    \includegraphics[width = 1.5in]{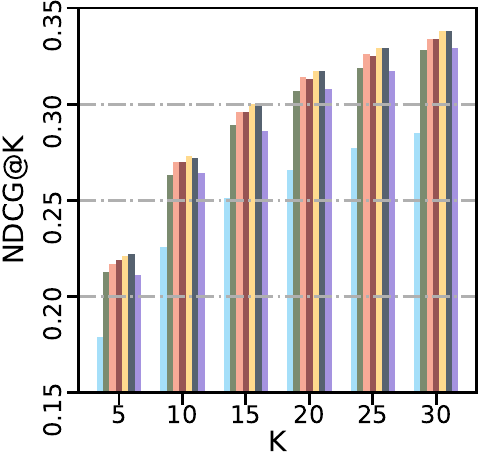}\vspace{-0.1em}
    &\includegraphics[width = 1.48in]{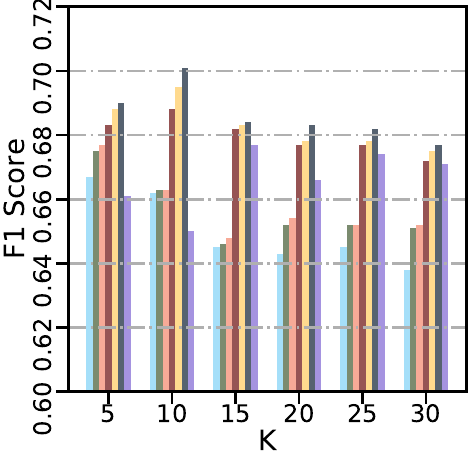}\vspace{-0.1em}
    &\includegraphics[width = 1.48in]{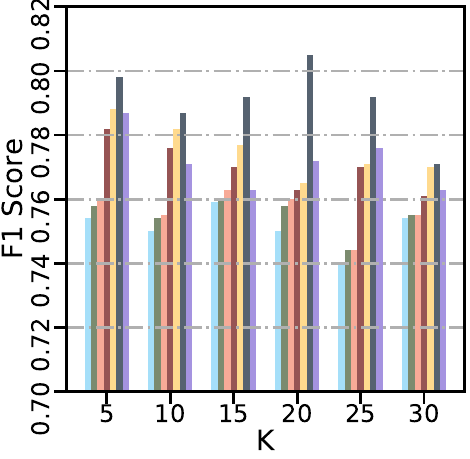}\vspace{-0.1em}
    &\includegraphics[width = 1.48in]{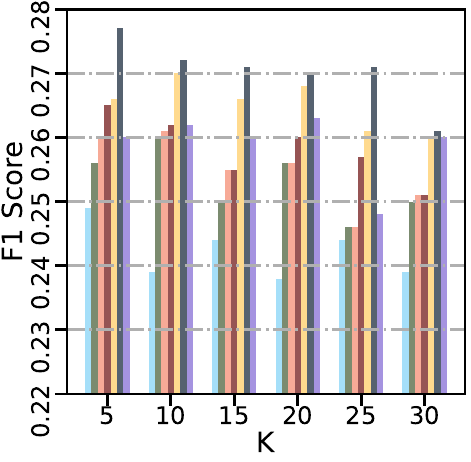}\vspace{-0.1em}\\
    \multicolumn{4}{c}{\small{(a) Above: Recommendation, Age, Gender and Occupation Inference  results w.r.t. privacy budget $\epsilon$}}\vspace{0.25em}\\
    
	\multicolumn{4}{c}{\includegraphics[scale=0.60]{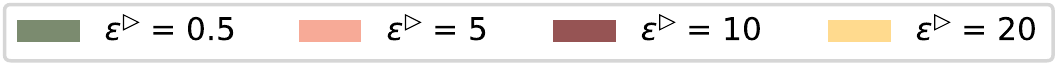}}\\
	\includegraphics[width = 1.5in]{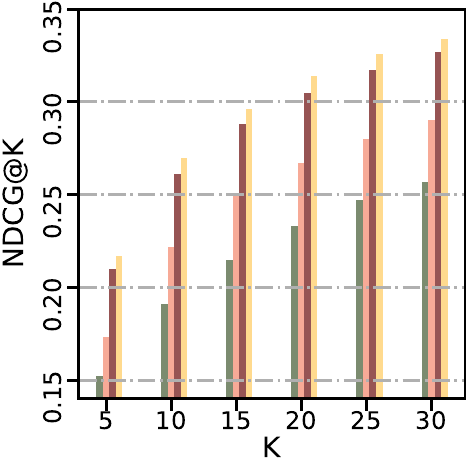}\vspace{-0.1em}&
    \includegraphics[width = 1.46in]{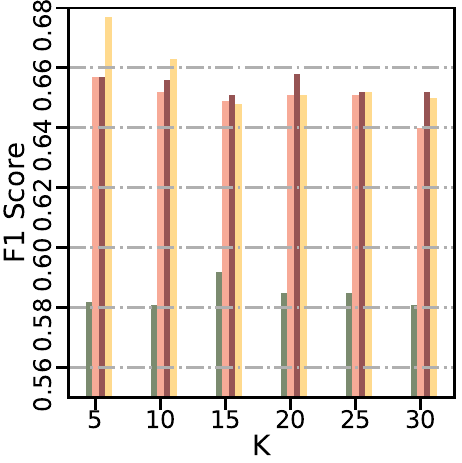}\vspace{-0.1em}
    &\includegraphics[width = 1.46in]{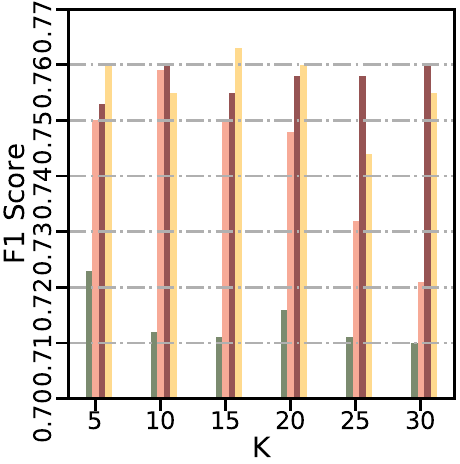}\vspace{-0.1em}
    &\includegraphics[width = 1.46in]{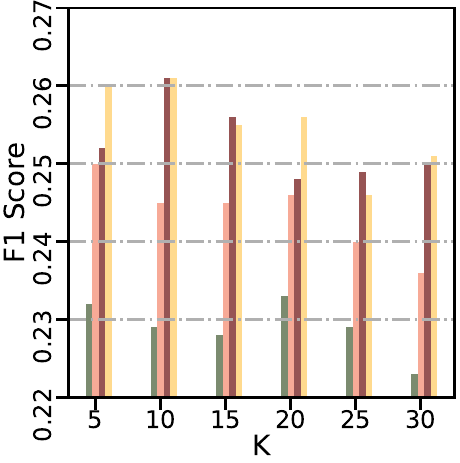}\vspace{-0.1em}\\
    \multicolumn{4}{c}{\small{(b) Above: Recommendation, Age, Gender and Occupation Inference  results w.r.t. privacy budget $\epsilon^{\rhd}$}}\vspace{0.25em}\\

	\multicolumn{4}{c}{\includegraphics[scale=0.60]{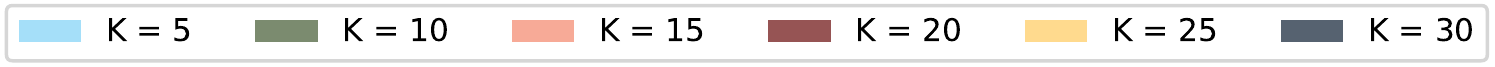}}\\
	\includegraphics[width = 1.52in]{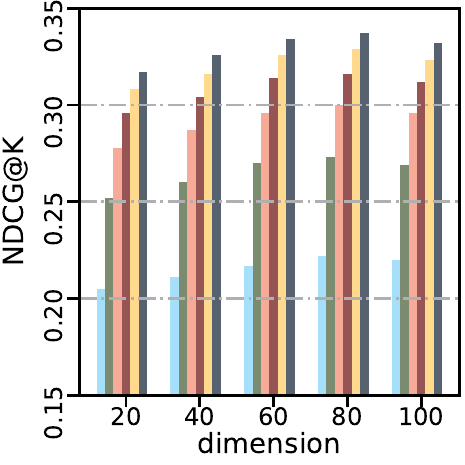}\vspace{-0.1em}&
    \includegraphics[width = 1.48in]{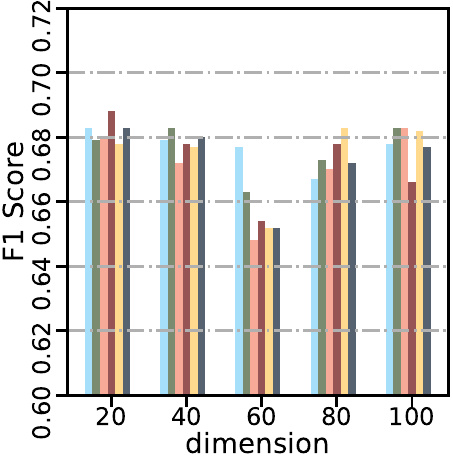}\vspace{-0.1em}
    &\includegraphics[width = 1.48in]{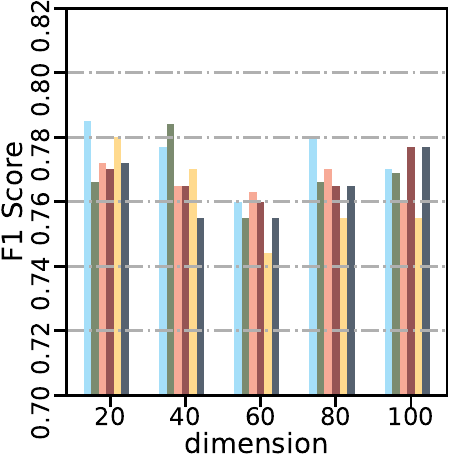}\vspace{-0.1em}
    &\includegraphics[width = 1.48in]{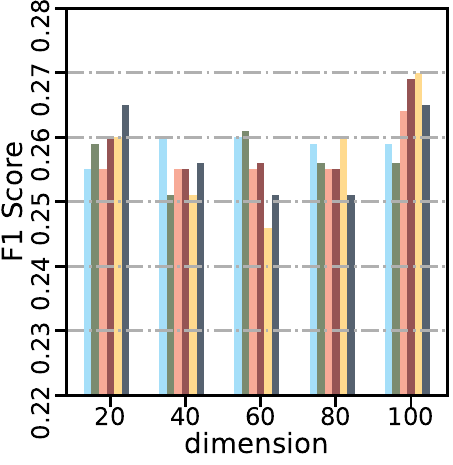}\vspace{-0.1em}\\
  \multicolumn{4}{c}{\small{(c) Above: Recommendation, Age, Gender and Occupation Inference  results w.r.t. dimension $d$}}\\
	\end{tabular}
\vspace{-0.4cm}
\caption{Recommendation and privacy protection results w.r.t. privacy budget $\epsilon$, $\epsilon^{\rhd}$ and $d$.}
\label{fig:para_epsilon}
\vspace{-0.8em}
\end{figure*}
Table~\ref{table:attack} shows the F1 scores achieved by the attribute inference attack model described in Section~\ref{sec:protocol} on all the baselines. Lower F1 scores show higher resistance of the recommender to attribute inference attacks. Obviously, GERAI constantly outperforms all baselines with $K \in \{ 15, 20, 25, 30\}$, indicating that our model is able to protect users' privacy and produce recommendations with strong privacy guarantee. Though RAP achieves slightly better results on the age attribute at $K=5$ and $K=10$, it falls behind GERAI in all other cases. As a model specifically designed for supervised learning, RAP is naturally robust against attribute inference attack. 
We also observe that GERAI has significantly better performance against attribute inference attack in comparison to Blurm that obfuscates user-item rating data to the recommender system. The results confirm the effectiveness of our dual-stage perturbation in private attribute protection. In addition, compared with conventional recommender systems that collaboratively model user-item interactions (i.e., BPR and GCN), models that make use of differential privacy (i.e., DPAE, DPMF, DPNE and GERAI) show obvious superiority in resistance to attribute inference attack. However, compared with all DP-based recommender systems, GERAI achieves significantly lower $F1$ score for all three private attributes and thus outperform those methods in terms of obscuring users' private attribute information. The reason is that the proposed privacy mechanisms in those DP-based methods cannot have the same strength as GERAI on preventing leakage of sensitive information from recommendation results. This further validates that incorporating differential privacy may prevent directly disclosing private attributes, but these methods cannot effectively provide higher privacy levels. Furthermore, with the increasing value of $K$, the performance of the attacker slightly decreases. One possible reason is that, more recommended products will become a natural ``noise'' to help reduce the risk of privacy disclosure. Finally, we observe that GCN has the weakest privacy protection results because it directly incorporates the node features with sensitive information. Note that compared with GCN, GERAI achieves an average relative improvement of $
11\%$, $14.4\%$ and $6.75\%$ respectively on age, gender and occupation, which implies that DP can ensure that the published recommendations of GERAI can avoid breaching users' privacy. 
\vspace{-1em}
\subsection{Recommendation Effectiveness (RQ2)}
We summarize all models' performance on personalized recommendation with Table~\ref{table:rec}. Note that higher $Hit@K$ and $NDCG@K$ values imply higher recommendation quality. Firstly, GERAI outperforms all privacy-preserving baselines consistently in terms of both $Hit@K$ and $NDCG@K$. Particularly, the improvement of GERAI with $K = 5$ demonstrate that our model can accurately rank the ground truth movies at the top-$5$ positions. In addition, compared with RAP, GERAI yields recommendation results that are closer to the 
state-of-the-art GCN. Thanks to the dual-stage perturbation setting where two sets of privacy budgets are used, a relatively higher privacy for user feature perturbation does not significantly impede the recommendation accuracy, and is sufficient for high-level attribute protection. Furthermore, the gap between the ranking accuracy drops with the increasing value of $K$. Finally, GCN achieves the best performance among all methods except when $K = 30$, which showcases the intrinsic strength of GCN-based recommenders. Meanwhile, Blurm has the worst performance among all methods as the way it adds noise to the user-item interaction data is harmful for the recommendation quality. 
\vspace{-1em}
\subsection{Accuracy and Privacy (RQ3)}\label{sec:hyper_analysis}
We answer RQ3 by investigating the performance fluctuations of GERAI with varied global and local privacy budgets $\epsilon$, $\epsilon^{\rhd}$ and embedding dimension $d$. We vary the value of one hyperparameter while keeping the other unchanged, and record the new recommendation and attribute inference results achieved. Figure~\ref{fig:para_epsilon} plots the results with different parameter settings.

\textbf{Impact of Global Privacy Budget $\epsilon$ for Loss Perturbation.} The value of the privacy budget $\epsilon$ is examined in $\{0.1, 0.2, 0.4, 0.8, 1.6,\\ 3.2\}$. In general, our GERAI outperforms RAP in terms of recommendation accuracy, and the performance improvement tends to become less significant when $\epsilon$ becomes quite small. Since a smaller $\epsilon$ requires a larger amount of noise to be injected to the objective function, it negatively influences the recommendation results. The results further confirms the effectiveness of GCNs-based recommendation component in our model, which helps GERAI preserve recommendation quality in practice. Furthermore, though the attack results illustrate that a relatively small $\epsilon$ (large noise) can obtain better performance on privacy protection within our expectation, it also degraded recommendation results correspondingly. Compared with RAP, the results imply that, by choosing a proper value of $\epsilon$ ($0.4$ in our case), our GERAI can achieve a good trade-off between privacy protection and recommendation accuracy.

\textbf{Impact of Local Privacy Budget $\epsilon^{\rhd}$ for User Feature Perturbation.} We study the impact of the privacy budget on input features with $\epsilon^{\rhd} \in \{0.5, 5, 10, 20\}$. It is worth mentioning that we seek a relatively higher value of $\epsilon^{\rhd}$ to maintain moderate utility of user features. From Figure ~\ref{fig:para_epsilon}, we can draw the observation that though reducing the value of privacy budget $\epsilon^{\rhd}$ in the input features may help the model yield better performance against attribute inference attack, GERAI generally achieves a significant drop on recommendation performance with a smaller $\epsilon^{\rhd}$. Particularly, when $\epsilon^{\rhd} = 0.5$, the recommendation results show that GERAI cannot capture users' actual preferences. This is because the feature vector $\widehat{\textbf{x}}_u$ determines the number of non-zero elements in base embedding of our model, which can cause significant information loss when it is small. 
As the recommendation is also highly accurate when $\epsilon^{\rhd} = 10$, the attribute inference performance achieved by the attacker is occasionally comparable to setting $\epsilon^{\rhd} = 20$. Overall, setting $\epsilon^{\rhd}$ to $20$ is sufficient for preventing privacy leakage, while helping GERAI to achieve optimal recommendation results.

\textbf{Impact of Dimension $d$.} As suggested by Eq.(\ref{eq:Delta}), the dimension $d$ controls the privacy sensitivity $\Delta$ and our model's expressiveness of the network structure. We vary the dimension $d$ in $\{20, 40, 60, 80 , 100\}$ and the corresponding noise parameters in Lapla-ce distribution are $\{0.00375, 0.01375, 0.03, 0.05, 0.08\}$. Obviously, the recommendation accuracy of GERAI benefits from a relatively larger dimension $d$, but the privacy protection performance is not always lower with a large $d$. The reason is that the value of the dimension $d$ is directly associated with our model's expressiveness, which means that a relatively larger $d$ can improve the recommendation results, providing better inputs to the attacker model as well. Furthermore, as shown in Figure~\ref{fig:para_epsilon}, the best privacy protection performance is commonly observed with $d=60$.
\vspace{-0.5em}
\subsection{Importance of Privacy Mechanism (RQ4)}\label{sec:ablation}
\begin{table}[t]
    \caption{Ablation test results.}
    \vspace{-0.5em}
    \centering
    \scalebox{0.82}{%
    \begin{tabular}{|p{1.6cm}<{\centering} |p{1.3cm}<{\centering}|p{1.3cm}<{\centering}|p{1.3cm}<{\centering}|p{1.3cm}<{\centering}|p{1.3cm}<{\centering}|}
        \hline
         \multirow{2}{*}{Variant}&\multicolumn{2}{c|}{Recommendation Task}&\multicolumn{3}{c|}{Attribute Inference Attack (F1 score)}\\
          \cline{2-6}
         &Hit@5&NDCG@5&Age&Gen&Occ  \\
         \hline
         GCN&0.365&0.247&0.697&0.851&0.277\\
         GERAI-NL&0.340&0.221&0.688&0.791&0.270\\
         GERAI-NF&0.337&0.219&0.679&0.788&0.266\\
         \hline
         GERAI&0.333&0.217&0.677&0.760&0.260\\
         \hline
    \end{tabular}}
    \label{tab:my_label}
  \vspace{-0.5em}
\end{table}
To better understand the performance gain from the major components proposed in GERAI, we perform ablation analysis on different degraded versions of GERAI. Each variant removes one privacy mechanism from the dual-stage perturbation paradigm. Table~\ref{tab:my_label} summarizes the outcomes in two tasks in terms of $Hit@5$, $NDCG@5$ and \textit{F1 score}. For benchmarking, we also demonstrate the results from the full version of GERAI and the non-private GCN.

\textbf{Removing perturbation at input stage (GERAI-NL).} The GERAI-NL only enforces $\epsilon$-differential privacy by perturbing the objective function in Eq. (\ref{eq:loss_approx}). We remove the privacy mechanism in users' features by sending raw features $\mathbf{X}$ directly into the recommendation component. After that, a slight performance decrease in the recommendation accuracy appeared, while achieving better performance against attribute inference attack. The results confirm that the functional mechanism in our model can help a GCN-based recommender satisfy privacy guarantee and yield comparable recommendation accuracy. In addition, GERAI significantly outperform GERAI-NL against attribute inference attack. Apparently, the raw user features are not properly perturbed in GERAI-NL, leading to a high potential risk in privacy leakage.   

\textbf{Removing perturbation at optimization stage (GERAI-NF).} We remove the privacy mechanism in objective function by setting $\epsilon = 0$. As the users' features are perturbed against information leaks, GERAI-NF achieves a significant performance improvement in the privacy protection, compared with the pure GCN. In addition, the slight performance difference between GERAI and GERAI-NF in two tasks could be attributed to the perturbation strategy in objective function. It further verifies that the joint effect of perturbation strategies in objective function and input features are beneficial for both recommendation and privacy protection purposes.

\subsection{Robustness against Different Attribute Inference Attackers (RQ5)}
\begin{table}[t]
\centering
\caption{Performance of attribute-inference attack w.r.t. different types of attacker.}
\vspace{-0.5em}
\scalebox{0.90}{%
 \begin{tabular}{|p{1.2cm}<{\centering} | p{1.2cm}<{\centering}| p{1.2cm}<{\centering} |p{1.2cm}<{\centering} |p{1.2cm}<{\centering}| p{1.2cm}<{\centering}|} 
 \hline
 \multirow{2}{*}{Attribute}&\multirow{2}{*}{Method}&\multicolumn{4}{c|}{F1 Score}\\
  \cline{3-6}
  &&DT&NB&KNN&GP\\
 \hline
 \multirow{7}{*}{Age}&BPR&0.466&0.376&0.487&0.260\\
 &GCN&0.513&0.366&0.487&0.619\\
 &Blurm&0.471&0.402&0.492&0.265\\
 &DPAE&0.492&0.481&0.476&\textbf{0.249}\\
 &DPNE&0.593&0.402&0.476&0.349\\
 &DPMF&0.561&0.402&0.486&0.275\\
 &RAP&0.476&0.407&0.513&0.534\\
 \cline{2-6}
 &GERAI&\textbf{0.434}&\textbf{0.365}&\textbf{0.466}&0.286\\
 \hline
 \hline
 \multirow{7}{*}{Gen}
 &BPR&0.651&0.444&0.561&0.672\\
 &GCN&0.635&0.429&0.566&0.810\\
 &Blurm&0.630&\textbf{0.370}&0.556&0.693\\
 &DPAE&0.635&0.381&\textbf{0.545}&0.683\\
 &DPNE&0.640&0.381&0.556&0.667\\
 &DPMF&0.683&0.376&0.556&0.683\\
 &RAP&0.670&0.439&0.619&0.709\\
 \cline{2-6}
 &GERAI&\textbf{0.619}&0.429&0.556&\textbf{0.656}\\
 \hline
 \hline
 \multirow{7}{*}{Occ}
 &BPR&0.132&0.148&0.070&0.116\\
 &GCN&0.122&0.148&0.063&0.222\\
 &Blurm&0.127&0.185&0.074&0.105\\
 &DPAE&0.111&0.180&\textbf{0.063}&0.212\\
 &DPNE&0.132&0.175&0.079&0.106\\
 &DPMF&0.175&0.180&0.074&0.104\\
 &RAP&0.122&0.148&0.090&0.127\\
 \cline{2-6}
 &GERAI&\textbf{0.111}&\textbf{0.116}&0.069&\textbf{0.090}\\
 \hline
 \end{tabular}}
 \label{table:attack_type}
\end{table}
In real-life scenarios, the models used by attribute inference attacker are usually unknown and unpredictable, so hereby we investigate how GERAI and other baseline methods perform in the presence of different types of attack models, namely Decision Tree (DT), Naive Bayesian (NB), KNN and Gaussian Process (GP), that are widely adopted classification methods. In this study, we use the top-5 recommendation generated by corresponding recommender methods for all attackers as introduced in Section~\ref{sec:protocol}. Table~\ref{table:attack_type} shows the attribute inference accuracy of each attacker. The first observation is that our proposed GERAI outperforms all the comparison methods in most scenarios. Though DPAE achieves slightly better results in several cases, its recommendation accuracy is non-comparable to GERAI. This further validates the challenge of incorporating privacy protection mechanism for personalized recommendation. Another observation is that there is a noticeable performance drop of RAP facing non-DNN attacker models. As RAP is trained to defend a specific DNN-based inference model, RAP is more effective when attacker is also DNN-based as shown in Table~\ref{table:attack}. However, RAP underperforms when facing the other five commonly used inference models, showing that GERAI can more effectively resist attribute inference attacks and protect users' privacy without any assumption on the type of attacker models.
\vspace{-0.5em}
\section{Related Work}
\textbf{Attribute Inference Attacks.} The target of attribute inference attack is inferring users' private attribute information from their publicly available information (e.g. recommendations). Three main branches of attribute inference attack approaches are often distinguished: friend-based, behavior-based and hybrid approaches. Friend-based approaches infer the target user's attribute in accordance with the target's friends' information~\cite{lindamood2009inferring,he2006inferring,gong2014joint}. He et al~\cite{he2006inferring} first constructed a Bayesian network to model the causal relations among people in social networks, which is used to obtain the probability that the user has a specific attribute. Behavior-based approaches achieve this purpose via users' behavioral information such as movie-rating behavior~\cite{weinsberg2012blurme} and Facebook likes~\cite{kosinski2013private}. The third type of works exploits both friend and behavioral information~\cite{jia2017attriinfer,gong2018attribute,gong2016you}. For example,~\cite{gong2014joint} creates a social-behavior-attribute network to infer attributes. Another work~\cite{jia2017attriinfer} 
models structural and behavioral information from users who do not have the attribute in the training process as a pairwise Markov Random Field. 

\textbf{Privacy and Recommender System.} With the growth of online platforms (e.g. Amazon), recommender systems play a pivotal role in promoting sales and enhancing user experience. The recommendations, however, may pose a severe threat to user privacy such as political inclinations via attribute inference attack. Hence, it is of paramount importance for system designers to construct a recommender system that can generate accurate recommendations and guarantee the privacy of users. Current researches that address vulnerability to privacy attacks often rely on providing encryption schemes~\cite{kim2018efficient,canny2002collaborative} and differential privacy~\cite{kandappu2014privacycanary}. Encryption-based methods enhance privacy of the conventional recommender systems with advanced encryption techniques such as homomorphic encryption~\cite{kim2018efficient,chai2020secure}. However, these methods are considered computation expensive as a third-party crypto-service provider is required. DP-based recommender systems can provide a strong and mathematically rigorous privacy guarantee~\cite{mcsherry2009differentially,berlioz2015applying,liu2015fast}. Works in this area aim to ensure that the recommender systems are not sensitive to any particular record and thus prevent adversaries from inferring a target user's ratings. ~\cite{parra2014optimal} proposes a perturbation method that adds or removes items and ratings to minimize privacy risk. Similarly, RAPPOR~\cite{erlingsson2014rappor} is proposed to perturb the user's data before sending them to the server by using the randomized response.
More recently, graph embedding techniques have been opening up more chances to improve the efficiency and scalability of the existing recommender systems~\cite{cenikj2020boosting,ying2018graph}. As the core of GCN is a graph embedding algorithm, our work is also quite related to another area: privacy preservation on graph embedding. Hua et al.~\cite{jingyu1763differentially} and Shin et al.~\cite{shin2018privacy} proposed gradient perturbation algorithms for differentially private matrix factorization to protect users' ratings and profiles. Another work enforces differential privacy to construct private covariance matrices to be further used by recommender~\cite{dwork2006calibrating}. Liu et al.~\cite{liu2019differentially} proposed DPAE that leverages the privacy problem in recommendation with the Autoencoders. Gaussian noise is added in the process of gradient descent. 
However, the existing privacy-preserving works in recommendation systems focus on protecting users against the membership attacks in which an adversary tries to infer a targeted user's actual ratings and deduce if the target is in the database, which is not fulfilled in our scenario. These limitations motivated us to propose GERAI that is able to counter private attribute inference attacks in the personalized recommendation system. 

\section{Conclusion}
In this paper, we propose a GCN-based recommender system that guards users against  attribute inference attacks while maintaining utility, named GERAI. GERAI firstly masks users' features including sensitive information, and then incorporates differential privacy into the GCN, which effectively bridges user preferences and features for generating secure recommendations such that a malicious attacker cannot infer their private attribute from users' interaction history and recommendations. The experimental results evidence that GERAI can yield superior performance on both recommendation and attribute protection tasks. 

\section{Acknowledgments}
The work has been supported by Australian Research Council (Grant
No.DP190101985 and DP170103954).
\balance
\bibliographystyle{ACM-Reference-Format}
\bibliography{sample-base}

\end{document}